\documentclass[10pt,conference]{IEEEtran}

\usepackage{amsmath,amssymb}
\usepackage{subfigure}
\usepackage{graphicx,graphics,color,psfrag}
\usepackage{cite,balance}
\usepackage{caption}
\allowdisplaybreaks
\usepackage{algorithm}
\usepackage{algorithmic}
\usepackage{accents}
\usepackage{amsthm}
\usepackage{bm}
\usepackage{url}
\usepackage[english]{babel}
\usepackage{multirow}
\usepackage{enumerate}
\usepackage{cases}
\usepackage{stfloats}
\usepackage{dsfont}
\usepackage{color,soul}
\usepackage{amsfonts}
\usepackage{cite,graphicx,amsmath,amssymb}
\usepackage{subfigure}
\usepackage{fancyhdr}
\usepackage{hhline}
\usepackage{graphicx,graphics}
\usepackage{array,color}
\usepackage{mathtools}
\usepackage{amsmath}

\newtheorem{theorem}{\emph{\underline{Theorem}}}

\newtheorem{lemma}{\emph{\underline{Lemma}}}
\newtheorem{corollary}{\emph{\underline{Corollary}}}

\newtheorem{remark}{\bf \emph{\underline{Remark}}}

\def\({\left(}
\def\){\right)}

\def\b0{{\mathbf{0}}}

\newcommand{\nn}{\nonumber}

\usepackage[
top    = 1.65 cm,
bottom = 1.06 in,
left   = 0.6 in,
right  = 0.6 in]{geometry}

\IEEEoverridecommandlockouts\IEEEpubid{\makebox[\columnwidth]{ 978-1-6654-3540-6/22~\copyright~2022 IEEE \hfill} \hspace{\columnsep}\makebox[\columnwidth]{ }}

\begin{document}
\captionsetup[figure]{name={Fig.}}

\title{\huge 
Multi-Hop Beam Routing for Hybrid Active/Passive \\ IRS Aided Wireless Communications} 
\author{\IEEEauthorblockN{Yunpu Zhang and Changsheng You}
\IEEEauthorblockA{Department of Electronic and Electrical Engineering\\
Southern University of Science and Technology\\
Email: zhangyp2022@mail.sustech.edu.cn, youcs@sustech.edu.cn}}

\maketitle

\begin{abstract} 
Prior studies on intelligent reflecting surface (IRS) have mostly considered wireless communication systems aided by a \emph{single passive} IRS, which, however, has limited control over wireless propagation environment and suffers from product-distance path-loss. To address these issues, we propose in this paper a new \emph{hybrid active/passive} IRS aided wireless communication system, where an \emph{active} IRS and multiple passive IRSs are deployed to assist the communication between a base station (BS) and a remote user in complex environment, by establishing a multi-hop reflection path across active/passive IRSs. In particular, the active IRS enables signal reflection with power amplification, thus effectively compensating the severe path-loss in the multi-reflection path. To maximize the achievable rate at the user, we first design the optimal  beamforming of the BS and selected (active/passive) IRSs for a given multi-reflection path, and then propose an efficient algorithm to obtain the optimal multi-reflection path by using the path decomposition method and graph theory.
We show that the active IRS should be selected to establish the beam routing path when its amplification power and/or number of active reflecting elements are sufficiently large. Last, numerical results demonstrate the effectiveness of the proposed hybrid active/passive IRS beam routing design as compared to the benchmark scheme with passive IRSs only.
\end{abstract}
\vspace{-8pt}
\section{Introduction}

Intelligent reflecting surface (IRS) has emerged as a promising technology to smartly reconfigure the wireless propagation environment by dynamically tuning its reflecting elements \cite{wu2021intelligent}. This thus has motivated substantial research recently to incorporate IRS into traditional wireless systems for improving the communication performance (see, e.g., \cite{liu2021reconfigurable,9690635,huang2020holographic,9724202}). 

In the existing works on IRS, most of them have considered the basic setup of the \emph{single passive} IRS aided wireless system. However, employing only a single IRS usually has limited control over the wireless channels, hence may not be able to unlock the full potentials of IRS. For example, there may not exist a blockage-free reflection link between the base station (BS) and a remote user via a single IRS in complex environment. To address this issue, new research efforts have been recently devoted to designing efficient multi-IRS aided systems by deploying two or more IRSs in the network that collaboratively establish multi-hop signal reflection paths from the BS to the users for bypassing the obstacles in between (see, e.g., \cite{you2020wireless, mei2020cooperative, mei2021intelligent}). 
Despite the prominent multiplicative passive beamforming gain, the multi-passive-IRS aided system suffers from severe \emph{product-distance} path-loss arising from the multi-hop reflection.
\begin{figure}[t]
    \centering
    \includegraphics[width=7cm]{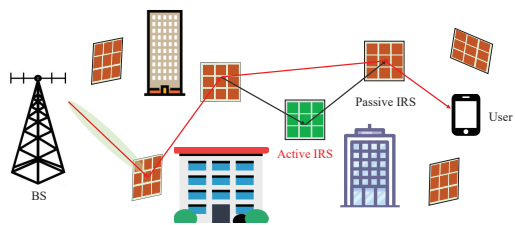}
    \caption{Hybrid active/passive IRS aided wireless communication systems.}\label{fig:system}
    \vspace{-22pt}
\end{figure}

To deal with the above issue, a new IRS architecture, called \emph{active} IRS, has been recently proposed \cite{long2021active,zhang2021active,you2021wireless,zeng2021active}. Specifically, with low-power reflection-type amplifiers equipped (e.g., tunnel diode and negative impedance converter), the active IRS is capable of reflecting incident signals with \emph{power amplification} at full-duplex mode, thus effectively compensating the product-distance path-loss at modest energy and hardware cost. Note that the active IRS generally outperforms the conventional amplify-and-forward (AF) relay since the latter operates at half-duplex mode and thus suffers lower spectrum efficiency.  These thus motivate the current work to incorporate the active IRS into the wireless system aided by passive IRSs only, so as to reap both the active-IRS amplification power gain and the multiplicative IRS beamforming gain. However, it remains unknown how to leverage the active and passive IRSs for constructing effective multi-hop beam routing paths and whether the newly added active IRS can bring significant performance gain.

To answer the above questions, we propose in this paper a new \emph{hybrid active/passive} IRS aided wireless communication system as shown in Fig.~1, where an active IRS is added in a multi-passive-IRS aided system to assist the  single-user downlink communication via \emph{cooperative beam routing}. To maximize the achievable rate of the user, we first design the optimal  beamforming of the BS and selected (active/passive) IRSs for a given multi-reflection path, and then propose an efficient algorithm to obtain the optimal multi-reflection path by using the path decomposition method and graph theory. We show that the active IRS should be selected to establish the beam routing path when its amplification power and/or number of active reflecting elements are sufficiently large. Last, numerical results demonstrate the effectiveness of the proposed hybrid active/passive IRS beam routing design as compared to the conventional system with passive IRSs only.
\vspace{-5pt}
\section{System Model}
\vspace{-4pt}
As shown in Fig. 1, we consider a hybrid active/passive IRS aided wireless  communication system, where one active IRS (with $N$ elements) and multiple passive IRSs (each with $M$ elements) are deployed to assist the downlink communication from a $T$-antenna BS to a single-antenna user.\footnote{The extensions to the cases with multiple active IRSs and/or multiple users are more complicated, which are left for future work.} We assume that the direct BS$\to$user channel is severely blocked by scattered obstacles, hence the BS can  communicate with the user via a multi-reflection path only that is formed by a set of selected IRSs. We denote by $\bar{\mathcal{J}}\triangleq\{0, \mathcal{J}, J+1\}$ the set of all nodes, where nodes $0$ and $J + 1$ refer to the BS and user, respectively; and  $\mathcal{J}\triangleq\{1,2, \cdots, J\}$ denotes the set of active/passive IRSs with the active IRS indexed by $\ell\in \mathcal{J}$. Moreover, the number of elements/antennas of each node is denoted by $U_j, \forall j\in \bar{\mathcal{J}}$, e.g., $U_0=T$, $U_{\ell}=N$.

\underline{\bf Channel model:} Let $\boldsymbol{H}_{0, j}$, $j \in {\mathcal{J}}$  denote the channel from the BS to the active/passive IRS $j$; $\boldsymbol{H}_{i, j}$ with $i, j \in {\mathcal{J}}, i\neq j$ represent the channel from IRS $i$ to IRS $j$; and $\boldsymbol{h}_{i, J+1}^{H}$, $i \in \{0, \mathcal{J}\}$ denote the channel from the BS/IRS to the user, where the dimension of each channel depends on the number of antennas/elements of the corresponding nodes. The above channels are modeled as follows.

Similar to \cite{mei2020cooperative}, we consider the LoS links only for the IRS reflection design, while treating all non-LoS (NLoS) links as part of the environment scattering, which has been shown to have a marginal effect on the user performance, especially in high-frequency bands due to the much higher path-loss in NLoS channels \cite{mei2021intelligent}. As such, we denote by $a_{i,j}$, $i,j\in \bar{\mathcal{J}}$ the channel state indicator, where $s_{i,j}=1$ means that there exists a LoS link between nodes $i$ and $j$, and $s_{i,j}=0$ otherwise. The LoS availability for all links is assumed to be known \emph{a priori} by using e.g., offline/online beam training methods \cite{you2020fast,mei2021intelligent}. As such, if $s_{i,j}=1$ for $i,j \in {\mathcal{J}}, i\neq j$, the inter-IRS channel, $\boldsymbol{H}_{i, j}$, can be modeled as
\begin{align}
\boldsymbol{H}_{i, j}=h_{i,j}\boldsymbol{a}_{\rm r}\left(\vartheta^{\rm r}_{i,j}, \theta^{\rm r}_{i,j}, U_j\right) \boldsymbol{a}_{\rm t}^H \left(\vartheta^{\rm t}_{i,j}, \theta^{\rm t}_{i,j}, U_i\right),\nn
\vspace{-7pt}
\end{align}
where $h_{i,j}=\frac{\sqrt{\beta}}{d_{i, j}} e^{-\frac{\jmath 2 \pi d_{i, j}}{\lambda}}$ denotes the complex channel gain of the link with $d_{i,j}$ denoting the link distance, $\beta$ denoting the reference channel gain at a distance of $1$ meter (m), and $\lambda$ denoting the carrier wavelength.  Moreover, $\vartheta^{\rm r}_{i,j}$ (or $\theta^{\rm r}_{i,j})$  denotes the azimuth (or elevation) angle-of-arrival (AoA) at node $j$ from node $i$, $\vartheta^{\rm t}_{i,j}$ (or $\theta^{\rm t}_{i,j})$ denotes the azimuth (or elevation) angle-of-departure (AoD) from node $i$ to node $j$; and $\boldsymbol{a}_{\rm r}$ and $\boldsymbol{a}_{\rm t}$ denote respectively the receive and transmit steering vectors. Specifically, based on the uniform rectangular array (URA) model for the IRS, $\boldsymbol{a}_{\rm r}$ can be modeled as 
$\boldsymbol{a}_{\rm r}\left(\vartheta^{\rm r}_{i,j}, \theta^{\rm r}_{i,j}, U_j\right)= \boldsymbol{u}(\frac{2 d_{\rm I}}{\lambda}\sin\theta^{\rm r}_{i,j} \cos\vartheta^{\rm r}_{i,j}, U_j^{(1)})\otimes \boldsymbol{u}(\frac{2 d_{\rm I}}{\lambda}\cos\theta^{\rm r}_{i,j}, U_j^{(2)})$, where $U_j^{(1)}$ and $U_j^{(2)}$ denote respectively the number of horizontal and vertical elements of node $j$, and the function $\boldsymbol{u}$ is defined as
$\boldsymbol{u}(\zeta, U)\triangleq [1, e^{-\jmath \pi \zeta}, \dots, e^{-\jmath (U-1) \pi \zeta}]$.
Similarly, the transmit steering vector, $\boldsymbol{a}_{\rm t}$, can be modeled by the same method. Besides, the BS$\to$IRS channel, $\boldsymbol{H}_{0, j}$, can be  modeled as
\begin{align}
\boldsymbol{H}_{0, j}=h_{0,j}\boldsymbol{a}_{\rm r}\left(\vartheta^{\rm r}_{0,j}, \theta^{\rm r}_{0,j}, U_j\right) \boldsymbol{a}_{\rm t}^H \left(\vartheta^{\rm t}_{0,j}, U_0\right),\nn
\end{align}
where  $\boldsymbol{a}_{\rm t}^H \left(\vartheta^{\rm t}_{0,j}, U_0\right)=\boldsymbol{u}^H(\frac{2 d_{\rm I}}{\lambda}\cos\theta^{\rm t}_{0,j}, U_0)$ is the transmit steering vector of the BS based on the uniform linear array (ULA) model. The BS/IRS$\to$user channel can be modeled as
\begin{align}
\boldsymbol{h}^H_{i, J+1}=h_{i,J+1}\boldsymbol{a}_{\rm t}^H \left(\vartheta^{\rm t}_{i,J+1}, \theta^{\rm t}_{i,J+1}, U_i\right).\nn
\end{align}

{\underline{\bf IRS model:}} Let $\boldsymbol{\Psi}_{j}=\operatorname{diag}\left(e^{\jmath \phi_{j, 1}}, e^{\jmath \phi_{j, 2}}, \cdots, e^{\jmath \phi_{j, M}}\right)$ denote the reflection matrix of the passive IRS $j \in \mathcal{J} \setminus \{\ell\}$, where for simplicity we set the reflection amplitude as one (i.e., its maximum value), and $\phi_{j, m}\in[0, 2\pi]$ represents the phase-shift at element $m \in \mathcal{M} \triangleq\{1, \cdots, M\}$. Besides,  the reflection matrix of the active IRS is denoted by  $\boldsymbol{\Psi}_{\ell}=\operatorname{diag}\left(\eta_{1}e^{\jmath \phi_{\ell, 1}}, \eta_{2}e^{\jmath \phi_{\ell, 2}}, \cdots, \eta_{N}e^{\jmath \phi_{\ell, N}}\right)$, where $\eta_{n}$ and $\phi_{\ell, n}$ represent the reflection amplitude and phase-shift at each active element $n \in \mathcal{N} \triangleq  \{1, \cdots, N\} $, respectively. 
Moreover, based on the LoS channel model, it can be shown that without loss of optimality, all active reflecting elements should adopt the same amplification factor, i.e., $\eta_{n}=\eta$, $\forall n \in \mathcal{N}$. Accordingly, $\boldsymbol{\Psi}_{\ell}$ can be equivalently expressed as $\boldsymbol{\Psi}_{\ell}=\eta\boldsymbol{\Phi}_{\ell}$, where $\boldsymbol{\Phi}_{\ell} \triangleq \operatorname{diag}\left(e^{\jmath \phi_{\ell, 1}}, e^{\jmath \phi_{\ell, 2}}, \cdots, e^{\jmath \phi_{\ell, N}}\right)$. It is worth noting that the active IRS incurs non-negligible thermal noise at all reflecting elements as compared to the passive IRS. The amplification noise is denoted by $\boldsymbol{n}_{\mathrm{F}} \in \mathbb{C}^{N \times 1}$, which is assumed to follow the independent circularly symmetric complex Gaussian distribution, i.e., $\boldsymbol{n}_{\mathrm{F}} \sim \mathcal{C} \mathcal{N}\left(\mathbf{0}_{N}, \sigma_\mathrm{F}^{2} \mathbf{I}_{N}\right)$ with $\sigma_\mathrm{F}^{2}$ denoting the amplification noise power.

\vspace{-5pt}
\section{Problem Formulation}

As the signal models for the IRS beam routing  with and without the active IRS differ significantly, we divide the IRS beam routing design into two cases and formulate their corresponding optimization problems.

\vspace{-5pt}
\subsection{Passive-IRS Beam Routing} 
First, consider the case where the active IRS is not involved in the beam routing. Let $\tilde{\Omega}=\left\{a_{1}, a_{2}, \ldots, a_{\tilde K}\right\}$ define the multi-reflection path from the BS to the user, where $\tilde{K}$ is the number of selected passive IRSs and $a_{\tilde{k}} \in \{\mathcal{J}\setminus \ell\}$, $\tilde{k} \in \tilde{\mathcal{K}} \triangleq  \{1, \cdots, \tilde{K}\}$ denote the index of $\tilde{k}$-th selected IRS.
Then the BS$\to$user equivalent channel is given by
\begin{equation}
\tilde{\boldsymbol {g}}_{\rm BU}^H=\boldsymbol{h}_{a_{\tilde K}, J+1}^{H} \boldsymbol{\Psi}_{a_{\tilde K}} \prod_{{\tilde k} \in \{\mathcal{\tilde K} \setminus \tilde K\}}\left(\boldsymbol{H}_{a_{{\tilde k}}, a_{{\tilde k}+1}} \boldsymbol{\Psi}_{a_{{\tilde k}}}\right) \boldsymbol{H}_{0, a_{1}}.
\end{equation}
As such, the received signal at the user via the multi-reflection path without the active IRS involved is given by
\vspace{-5pt}
\begin{equation}
y_{\mathrm{pas}}=\tilde{\boldsymbol {g}}_{\rm BU}^H   \boldsymbol{w}_{\rm B}x+n_0,
\vspace{-5pt}
\end{equation}
where $x$ denotes the transmitted signal with power $P_{\mathrm{B}}$, $\boldsymbol{w}_{\rm B} \in \mathbb{C}^{T\times 1}$ denotes the normalized beamforming vector of the BS with  $||\boldsymbol{w}_{\rm B} ||^{2}=1$, and $n_0$ denotes the received Gaussian noise at the user with power $\sigma^{2}$. The corresponding achievable rate in bits/second/Herz (bps/Hz) is given by
$R_{\mathrm{pas}}=\log_{2}\left(1+\gamma_{\mathrm{pas}}\right)$, where the received signal-to-noise-ratio (SNR) is \begin{equation}
\gamma_{\mathrm{pas}}=P_{\mathrm{B}} |\tilde{\boldsymbol{g}}_{\mathrm{BU}}^{H}  \boldsymbol{w}_{\rm B}|^{2}/\sigma^{2}.  
\vspace{-3pt}
\end{equation}

Note that the optimization problem for maximizing $R_{\mathrm{pas}}$ has been studied in \cite{mei2020cooperative}, for which the optimal beam routing design can be obtained by using a similar method in Section~\ref{Sec:Alg}, as detailed later.

\vspace{-6pt}
\subsection{Hybrid-IRS Beam Routing with Active IRS Involved} 
Next, we consider the case where the active IRS is assumed to involve in the beam routing. Let $\Omega=\left\{a_{1}, a_{2}, \ldots, a_{K}\right\}$ define the corresponding multi-reflection path from the BS to the user over $K$ active/passive IRSs,  where $a_{k} \in \mathcal{J}$, $k \in \mathcal{K} \triangleq  \{1, \cdots, K\}$ denotes the index of the $k$-th selected IRS. In particular,  the routing index of the active IRS $\ell$ is denoted as ${\mu(\ell)}$, i.e.,  $a_{\mu(\ell)}=\ell$. To facilitate the beam routing design in the sequel, we divide the multi-reflection path into two sub-paths. One is the reflection path from the BS to the active IRS, whose equivalent channel can be obtained as
\begin{equation}
\boldsymbol{G}_{\mathrm{BA}}=\prod_{k \in \{1, \cdots, \mu(\ell)-1\}}\left(\boldsymbol{H}_{a_{k}, a_{k+1}} \boldsymbol{\Psi}_{a_{k}}\right) \boldsymbol{H}_{0, a_{1}}.\label{Eq:GBA}
\vspace{-5pt}
\end{equation}
The other one is the reflection path from the active IRS to the user, which is given by
\begin{align}
\!\!\!\boldsymbol{g}^H_{\mathrm{AU}}&=\boldsymbol{h}_{a_{K}, J+1}^{H} \boldsymbol{\Psi}_{a_{K}} \nn\\
& \times \prod_{k \in \{\mu(\ell)+1, \cdots, K-1\}}\left(\boldsymbol{H}_{a_{k}, a_{k+1}} \boldsymbol{\Psi}_{a_{k}}\right) \boldsymbol{H}_{a_{\mu(\ell)}, a_{\mu(\ell)+1}}.\label{Eq:gAU}
\vspace{-5pt}
\end{align}
Based on the above, the received signal at the user via the multi-reflection path with the active IRS involved is given by
\begin{equation}
y_{\mathrm{act}}=\boldsymbol{g}_{\mathrm{AU}}^{H} \eta \boldsymbol{\Phi}_{\ell}\left(\boldsymbol{G}_{\mathrm{BA}}  \boldsymbol{w}_{\rm B}x+\boldsymbol{n}_{\mathrm{F}}\right)+n_0.
\vspace{-5pt}
\end{equation}
The corresponding achievable rate  is $R_{\mathrm{act}}=\log_{2}\left(1+\gamma_{\mathrm{act}}\right)$,
where the  received SNR is given by
\begin{equation}
\gamma_{\mathrm{act}}=\frac{P_{\mathrm{B}} |\boldsymbol{g}_{\mathrm{AU}}^{H} \eta \boldsymbol{\Phi}_{\ell} \boldsymbol{G}_{\mathrm{BA}} \boldsymbol{w}_{\rm B}|^{2}}{\|\boldsymbol{g}_{\mathrm{AU}}^{H} \eta \boldsymbol{\Phi}_{\ell}\|^{2} \sigma_{\mathrm{F}}^{2}+\sigma^{2}}.
\end{equation}
Note that the active IRS amplifies both the received signal and the noise at each active element. Let $P_{\mathrm{F}}$ denote the maximum amplification power of the active IRS, then we have
\begin{equation}
\eta^{2} (P_{\mathrm{B}}\|\boldsymbol{\Phi}_{\ell} \boldsymbol{G}_{\mathrm{BA}}  \boldsymbol{w}_{\rm B}\|^{2}+\sigma_{\mathrm{F}}^{2}\|\boldsymbol{\Phi}_{\ell} \mathbf{I}_{N}\|^{2}) \leq P_{\mathrm{F}}.\label{Eq:Power}
\end{equation}

Our goal is to maximize the achievable rate in the case of hybrid-IRS beam routing by optimizing the BS active beamforming, active/passive IRS beamforming, and hybrid-IRS beam routing. This problem can be  formulated as 
\begin{subequations}
\begin{align}
\!\!\!\max_{\substack{\boldsymbol{w}_{\rm B}, \{\boldsymbol{\Psi}_{a_{k}}\}, \boldsymbol{\Phi}_{\ell}, \eta, \Omega} }  &\log_{2} \left(1+\frac{ P_{\mathrm B}\left| \boldsymbol{g}_{\mathrm {AU}}^H \eta \boldsymbol{\Phi}_{\ell} \boldsymbol{G}_{\mathrm{BA}} \boldsymbol{w}_{\rm B} \right|^2 }{\| \boldsymbol{g}_{\mathrm{AU}}^H \eta \boldsymbol{\Phi}_{\ell}\|^{2} \sigma_{\mathrm{F}}^{2} + \sigma^{2}}\right)
\nn\\
\text{s.t.}~~~~~~
& \eqref{Eq:GBA},\eqref{Eq:gAU},\eqref{Eq:Power},\nn\\
& a_{k} \in \mathcal{J}, a_{k} \neq a_{k^{\prime}},~~\forall k,k^{\prime}\in\mathcal{K}, k \neq k^{\prime},\label{Eq:P1:1}\\
({\bf P1}):~~~~~~~~& s_{a_{k},a_{k+1}}=1,~~\forall k\in\mathcal{K}, k \neq K,\label{Eq:P1:2}\\
& s_{0,a_{1}}=s_{a_{K},J+1}=1,\label{Eq:P1:3}\\
& |[\boldsymbol{\Phi}_{\ell}]_{n,n}|=1,~~\forall n\in\mathcal{N}, \label{Eq:P1:4}
\\
& |[\boldsymbol{\Psi}_{a_{k}}]_{m,m}|=1,~~\forall m \in\mathcal{M}, k\in \mathcal{K} \setminus \{\mu(\ell)\}, \label{Eq:P1:5}
\end{align}
\end{subequations}
where $\boldsymbol{G}_{\mathrm{BA}}$ and $\boldsymbol{g}^H_{\mathrm{AU}}$ are given in \eqref{Eq:GBA} and \eqref{Eq:gAU}, respectively, which are  determined by the multi-reflection path $\Omega$ and active/passive IRS beamforming. Constraints described in \eqref{Eq:P1:1}--\eqref{Eq:P1:3} can ensure routing path $\Omega$ is feasible. 

Note that problem (P$1$) is a non-convex optimization problem due to the unit-modulus constraint
and  the intrinsic coupling between the two sub-paths (see, e.g.,  the objective and constraints \eqref{Eq:GBA}, \eqref{Eq:gAU}, \eqref{Eq:Power}). Hence, problem  (P$1$) is much more involved than that of the preceding passive-IRS case. To address this issue, an efficient method is proposed in the next section to obtain the optimal solution to problem (P$1$).

\section{Proposed Solution to Problem (P$1$)}

To solve problem (P$1$), we first optimize the joint beamforming design of the BS and IRSs for any given  multi-reflection path, based on which we further optimize the multi-reflection path. 
\vspace{-5pt}
\subsection{Joint Beamforming Optimization given Multi-reflection Path}
First, for any given feasible multi-reflection path with the active IRS involved (i.e., $\Omega$), problem (P$1$) is reduced to the following optimization problem for jointly designing the beamforming of the BS and active/passive IRSs.
\begin{subequations}
\begin{align}
\quad\ ({\bf P2}):\!\!\!\max_{\substack{\boldsymbol{w}_{\rm B}, \{\boldsymbol{\Psi}_{a_{k}}\}, \boldsymbol{\Phi}_{\ell}, \eta} }  &\log_{2} \left(1+\frac{ P_{\mathrm B}\left| \boldsymbol{g}_{\mathrm {AU}}^H \eta \boldsymbol{\Phi}_{\ell} \boldsymbol{G}_{\mathrm{BA}} \boldsymbol{w}_{\rm B} \right|^2 }{\| \boldsymbol{g}_{\mathrm{AU}}^H \eta \boldsymbol{\Phi}_{\ell}\|^{2} \sigma_{\mathrm{F}}^{2} + \sigma^{2}}\right)
\nn\\
\text{s.t.}~~~~~~
& \eqref{Eq:GBA},\eqref{Eq:gAU},\eqref{Eq:Power}, \eqref{Eq:P1:4}, \eqref{Eq:P1:5}.\nn
\end{align}
\end{subequations}
The optimal solution to problem (P$2$) is given as follows.
% This optimization problem can be To this end, we present a useful lemma as follows.
\begin{lemma}\label{Lem:GivenPath}
\emph{Given any feasible multi-reflection path $\Omega$, the optimal solution to problem (P$2$) is given by\begin{align}
\!\!\!\!\![\boldsymbol{\Psi}_{a_{k}}^*]_{m,m}&\!=\! \begin{cases}e^{\jmath (\angle [\underline{\boldsymbol{a}_{\rm t}}(a_{1})]_{m}-\angle[\underline{\boldsymbol{a}_{\rm r}}(a_{1})]_{m})} \!\!& \text { if } k=1 \\ e^{\jmath (\angle [\underline{\boldsymbol{a}_{\rm t}}(a_{K})]_{m}-\angle[\underline{\boldsymbol{a}_{\rm r}}(a_{K})]_{m})} & \text { if } k=K \\ e^{\jmath (\angle [\underline{\boldsymbol{a}_{\rm t}}(a_{k})]_{m}-\angle[\underline{\boldsymbol{a}_{\rm r}}(a_{k})]_{m})} & \text { otherwise }\end{cases},\!\!\!\label{Eq:PIRS}\\
%[\boldsymbol{\Psi}_{a_{k}}^*]_{m,m}&=e^{\jmath (\angle [\underline{\boldsymbol{a}_{\rm t}}(a_{k})]_{m}-\angle[\underline{\boldsymbol{a}_{\rm r}}(a_{k})]_{m})}, \forall a_{k} \in\Omega, m\in\mathcal{M},\label{Eq:PIRS}\\
\!\!\!\!\!\!\![\boldsymbol{\Phi}_{\ell}^*]_{n,n}&=e^{\jmath (\angle [\underline{\boldsymbol{a}_{\rm t}}(a_{\mu(\ell)})]_{n}-\angle[\underline{\boldsymbol{a}_{\rm r}}(a_{\mu(\ell)})]_{n})},\label{Eq:AIRS} \\
\boldsymbol{w}_{\rm B}^*&=\boldsymbol{a}_{\rm t} \left(\vartheta^{\rm t}_{0,a_{1}}, U_0\right)/\|\boldsymbol{a}_{\rm t} \left(\vartheta^{\rm t}_{0,a_{1}}, U_0\right)\|,\label{Eq:w}\\
\eta^*&=\sqrt{\frac{P_{\mathrm{F}}}{P_{\mathrm{B}}\|\boldsymbol{\Phi}^*_{\ell} \boldsymbol{G}_{\mathrm{BA}}  \boldsymbol{w}^*_{\rm B}\|^{2}+\sigma_{\mathrm{F}}^{2}\|\boldsymbol{\Phi}^*_{\ell} \mathbf{I}_{N}\|^{2}}}, \label{Eq:AF}
\end{align}}
\end{lemma}
\noindent where we denote $\underline{\boldsymbol{a}_{\rm t}}(a_{k}) \triangleq \boldsymbol{a}_{\rm t} (\vartheta^{\rm t}_{{a_{k}},{a_{k+1}}}, \theta^{\rm t}_{{a_{k}},{a_{k+1}}}, U_{a_{k}})$ and  $\underline{\boldsymbol{a}_{\rm r}}(a_{k}) \triangleq \boldsymbol{a}_{\rm r}(\vartheta^{\rm r}_{{a_{k-1}},{a_{k}}}, \theta^{\rm r}_{{a_{k-1}},{a_{k}}}, U_{a_{k}})$, in which $a_{0}\triangleq0$ (i.e., BS) and $a_{K+1}\triangleq J+1$ (i.e., user).

\begin{proof}
First, it can be shown that, the optimal phase-shift matrix of each passive IRS $a_{k}$ should be set as \eqref{Eq:PIRS}. Similarly, the optimal phase-shift matrix of the active IRS $\ell$ should align the cascaded BS$\to$active IRS$\to$user multi-reflection channel. Then, given any feasible multi-reflection path $\Omega$, the maximum-ratio transmission (MRT) is always the optimal BS active beamforming solution to problem (P$2$), i.e. \eqref{Eq:w}. Moreover, it can be shown that, at the optimal solution, the power constraint in \eqref{Eq:Power} is always active.
\end{proof} 

Lemma~\ref{Lem:GivenPath} shows that the optimal joint beamforming design of the BS and active/passive IRSs are closely coupled with the multi-reflection path of IRSs (see,  \eqref{Eq:PIRS}--\eqref{Eq:AF}). Specifically, the amplification factor of the active IRS is determined by the path-loss of the BS$\to$active IRS reflection link. The smaller the path-loss, the smaller the amplification factor. 
\vspace{-5pt}
\subsection{Multi-reflection Path Optimization}

Given the optimal beamforming design in Lemma \ref{Lem:GivenPath}, we have
\vspace{-5pt}
\begin{align}
f_{\rm BA}(\Omega)&\triangleq {\| \boldsymbol{G}_{\mathrm{BA}} \boldsymbol{w}_{\rm B}^*\|}^2=\frac{M^{2\mu(\ell)-2} T \beta^{\mu(\ell)}}{d_{0, a_{1}}^{2} \prod_{k=1}^{\mu(\ell)-1} d_{a_{k}, a_{k+1}}^{2}}, \label{Eq:BA}\\
f_{\rm AU}(\Omega)&\triangleq {\| \boldsymbol{g}^H_{\mathrm{AU}} \|}^2\nn\\
&=\frac{M^{2K-2\mu(\ell)} \beta^{K-\mu(\ell)+1}}{d_{a_{K}, J+1}^{2} d_{a_{\mu(\ell)}, a_{\mu(\ell)+1}}^{2} \prod_{k=\mu(\ell)+1}^{K-1} d_{a_{k}, a_{k+1}}^{2}},\label{Eq:AU}\\
\eta^2(\Omega)& =\frac{P_{\mathrm{F}}}{N (P_{\mathrm{B}} {f}_\mathrm{BA}(\Omega)+\sigma_{\mathrm{F}}^{2})},\label{Eq:eta}
\end{align}
where $f_{\rm BA}(\Omega)$ and $f_{\rm AU}(\Omega)$ represent respectively the end-to-end channel power gain of the BS$\to$active IRS and active IRS$\to$ user reflection sub-paths. 
Substituting \eqref{Eq:BA}--\eqref{Eq:eta} into problem (P$1$) yields the following equivalent problem that targets at optimizing the IRS multi-reflection path. 
\begin{subequations}
\begin{align}
({\bf P3}):\ \max_{\substack{ \Omega} }  ~~& \frac{P_{\mathrm{B}} N f_{\rm AU}(\Omega)f_{\rm BA}(\Omega)}{f_{\rm AU}(\Omega)  \sigma_{\mathrm{F}}^{2}+\frac{\sigma^{2}(P_{\mathrm{B}} f_{\rm BA}(\Omega)+\sigma_{\mathrm{F}}^{2})}{P_{\mathrm{F}}}}
\nn\\
\text{s.t.}~~
& \eqref{Eq:P1:1}-\eqref{Eq:P1:3}.\nn
\vspace{-3pt}
\end{align}
\end{subequations}
\addtolength{\topmargin}{0.02in}
Problem (P$3$) is still a non-convex optimization problem due to the coupling of the multi-reflection path in the objective and constraints of  \eqref{Eq:P1:1}--\eqref{Eq:P1:3}, which is thus generally difficult to solve. To deal with this challenge, we first show an interesting result that problem (P$3$) can be equivalently decomposed into two separate subproblems.

%In the following, we first decouple the problem (P3) into two mathematically tractable sub-problems, and then characterize the effects of the active IRS on the optimal beam routing design.

\begin{lemma}\label{Lem:Decouple}
\emph{The solution to problem (P$3$) can be obtained by solving the following two subproblems separately.
\begin{subequations}
\begin{align}
({\bf {P4.a}}):
%\  \left\{
%\begin{aligned}
\max_{\substack{ \Omega}_{\rm BA} } ~~& \frac{M^{2\mu(\ell)-2} T \beta^{\mu(\ell)}}{d_{0, a_{1}}^{2} \prod_{k=1}^{\mu(\ell)-1} d_{a_{k}, a_{k+1}}^{2}}
\nn\\
\text{s.t.}~~
& \eqref{Eq:P1:1}-\eqref{Eq:P1:3},\nn\\
({\bf {P4.b}}):\max_{\substack{ \Omega}_{\rm AU}  } ~~& \frac{M^{2K-2\mu(\ell)} \beta^{K-\mu(\ell)+1}}{d_{a_{K}, J+1}^{2} d_{a_{\mu(\ell)}, a_{\mu(\ell)+1}}^{2} \prod_{k=\mu(\ell)+1}^{K-1} d_{a_{k}, a_{k+1}}^{2}}
\nn\\
\text{s.t.}~~
& \eqref{Eq:P1:1}-\eqref{Eq:P1:3},\nn
%\end{aligned}
%\right
\end{align}
\end{subequations}
where $\Omega_{\rm BA}\triangleq\{a_1, a_2, \ldots, a_{\mu(\ell)-1}\}$ and $\Omega_{\rm AU}\triangleq 
\{a_{\mu(\ell)+1}, a_{\mu(\ell)+2}, \ldots, a_{K}\}$.}
\end{lemma}

\begin{proof}
First, we consider the effect of $f_{\rm AU}$ on the received SNR, given any fixed $f_{\rm BA}$. The received SNR is given by
\vspace{-3pt}
\begin{equation}\label{DE:1}
\gamma_{\mathrm{act}}^{(1)}= \frac{P_{\mathrm{B}} N f_{\rm BA}}{  \sigma_{\mathrm{F}}^{2}+\frac{\sigma^{2}(P_{\mathrm{B}} f_{\rm BA}+\sigma_{\mathrm{F}}^{2})}{P_{\mathrm{F}}f_{\rm AU}}}.
\vspace{-6pt}
\end{equation}
It can be easily shown from \eqref{DE:1} that $\frac{\partial \gamma_{\mathrm{act}}}{\partial f_{\rm AU}} > 0$, and hence maximizing $\gamma_{\mathrm{act}}^{(1)}$ is equivalent to maximizing $f_{\rm AU}$ and  independent of $f_{\rm BA}$. Second,  we consider only the effect of the $f_{\rm BA}$ given any  $f_{\rm AU}$. The corresponding received SNR is given by
\vspace{-3pt}
\begin{equation}\label{DE:2}
\gamma_{\mathrm{act}}^{(2)}= \frac{P_{\mathrm{B}}Nf_{\rm AU}}{\frac{f_{\rm AU}  \sigma_{\mathrm{F}}^{2}}{f_{\rm BA}}+\frac{\sigma^{2}(P_{\mathrm{B}} +\sigma_{\mathrm{F}}^{2}/f_{\rm BA})}{P_{\mathrm{F}}}},
\end{equation}
for which we can easily show that  $\frac{\partial \gamma_{\mathrm{act}}}{\partial f_{\rm BA}}>0$. This indicates that regardless of the multi-reflection path from the active IRS to user, it is always optimal to maximize $f_{\rm BA}$ for maximizing $\gamma_{\mathrm{act}}^{(2)}$.
These lead to the conclusion that maximizing the received SNR over the multi-reflection path can be equivalently transformed to the channel power gain maximization of the two sub-paths. Combining the above results with \eqref{Eq:BA} and \eqref{Eq:AU} leads to the desired result.
\end{proof}
\vspace{-3pt}
\begin{remark}[Decomposable multi-reflection path]\label{Rem:RP}
\emph{Lemma~\ref{Lem:Decouple} shows an interesting result that the hybrid-IRS beam routing design with the active IRS involved can be decomposed into two separate sub-path routing designs, which can be intuitively explained as follow. First, the BS$\to$active IRS reflection path determines the incident signal power at the active IRS. The larger the incident signal power, the smaller the effective received noise at the user (see \eqref{DE:2}) and hence a higher SNR. Second, the active IRS$\to$user reflection path determines the path-loss of the reflected signal. The smaller the path-loss, the larger the received signal power and hence a higher SNR.}
\end{remark}
\vspace{-3pt}
\begin{remark}[Reduced routing design complexity]\emph{Note that the original multi-reflection path design in problem (P$3$) is determined by all the $J$ IRSs, and the incoming and outgoing paths associated with the active IRS are intricately coupled in general. Specifically, the exhaustive search for the optimal routing has an exponential computational complexity. By using Lemma~\ref{Lem:Decouple}, the high-dimensional hybrid-IRS multi-reflection path design can be efficiently decoupled into two separate and lower-dimensional sub-path designs, leading to significantly reduced design complexity.} 
\end{remark}

Next, we solve the optimization problems  (P$4$.a) and  (P$4$.b) as follows by using graph theory  \cite{mei2020cooperative}.

\subsubsection{Problem reformulation}\label{Sec:Alg}

First, it can be shown that maximizing $f_{\rm BA}(\Omega)$ in problem (P$4$.a) and $f_{\rm AU}(\Omega)$ in problem (P$4$.b) are equivalent to minimizing 
\begin{align}\label{Eq:inv_G_BA}
\frac{1}{f_{\rm BA}(\Omega)}&=\frac{M^{2}}{T} \cdot \frac{d_{0, a_{1}}^{2}}{M^{2} \beta} \cdot  \prod_{k=1}^{\mu(\ell)-1} \frac{d_{a_{k}, a_{k+1}}^{2}}{M^{2} \beta},\\
\label{Eq:inv_g_AU}
\!\!\!\frac{1}{f_{\rm AU}(\Omega)}\!&=\!\frac{M^{2} d_{a_{\mu(\ell)}, a_{\mu(\ell)+1}}^{2}}{M^{2} \beta} \!\!\times\!\! \frac{d_{a_{K}, J+1}^{2}}{M^{2} \beta} \!\!\!\prod_{k=\mu(\ell)+1}^{K-1}\!\! \frac{d_{a_{k}, a_{k+1}}^{2}}{M^{2} \beta}.\!\!
\end{align}
Then, by taking the logarithm of \eqref{Eq:inv_G_BA} and \eqref{Eq:inv_g_AU} and ignoring irrelevant constant terms, problems  (P$4$.a) and (P$4$.b) can be reformulated as 
\vspace{-10pt}
\begin{align}
\label{P5:shortest}
({\bf P5.a}):~~
~\min _{\left\{a_{k}\right\}_{k=1}^{\mu(\ell)-1}, K} ~~&\ln \frac{d_{0, a_{1}}}{M \sqrt{\beta}}+\sum_{k=1}^{\mu(\ell)-1} \ln \frac{d_{a_{k}, a_{k+1}}}{M \sqrt{\beta}} 
\nn\\
\text{s.t.}~~
& \eqref{Eq:P1:1}-\eqref{Eq:P1:3}.\nn
\end{align}
\vspace{-20pt}
\begin{align}
({\bf P5.b}):
~\min _{\left\{a_{k}\right\}_{k=\mu(\ell)+1}^{K}, K} ~~& \ln \frac{d_{a_{\mu(\ell)}, a_{\mu(\ell)+1}}}{M \sqrt{\beta}}+\ln \frac{d_{a_{K}, J+1}}{M \sqrt{\beta}}\nn\\
&+\sum_{k=\mu(\ell)+1}^{K-1} \ln \frac{d_{a_{k}, a_{k+1}}}{M \sqrt{\beta}}
\nn\\
\text{s.t.}~~~
& \eqref{Eq:P1:1}-\eqref{Eq:P1:3}.\nn
\end{align}

\subsubsection{Modified shortest-path algorithm}
As problems (P$5$.a) and (P$5$.b) have similar forms, we only present the algorithm for solving (P$5$.a) in the sequel, while the same method can be used for solving problem (P$5$.b). To be specific, we first recast problem (P$5$.a) into a shortest simple-path problem (SSPP), by constructing a  directed weighted graph  $\mathcal{G}_{\mathrm{BA}}=(\mathcal{V}_{\mathrm{BA}}, \mathcal{E}_{\mathrm{BA}})$, where the vertex set is defined as $\mathcal{V}_{\mathrm{BA}}\triangleq\{0, \ell\} \cup \{\mathcal{J}_{\rm BA}\}$ (note that $\mathcal{J}_{\rm BA}$ is a set of nodes between the BS and the active IRS, and is determined by the network topology). Moreover, there exists an edge between any two vertexes $i,j \in \mathcal{V}_{\mathrm{BA}}$ if  $s_{i,j}=1$ and $d_{j,0} > d_{i,0}$, hence the edge set is expressed as $\mathcal{E}_{\mathrm{BA}}\triangleq\{(0,j)|s_{0,j}=1, j\in \mathcal{V}_{\mathrm{BA}}\} \cup \{(i,j)|s_{i,j}=1, d_{j,0}>d_{i,0}, i,j\in \mathcal{V}_{\mathrm{BA}}\} \cup \{(\ell-1,\ell)|s_{\ell-1,\ell}=1, \ell-1\in \mathcal{V}_{\mathrm{BA}}\}$. In addition, the weight between any two nodes $i,j \in \mathcal{V}_{\mathrm{BA}}$ is defined as $W_{i, j}=\ln (d_{i, j}/M \sqrt{\beta}) $. 

Then,  problem (P$5$.a) can be solved by using graph-optimization algorithms. It is worth noting that different from the conventional graph with positive edge weights only, there may exist negative edge weights in the constructed graph $\mathcal{G}_{\mathrm{BA}}$  (i.e., $d_{i, j}<M \sqrt{\beta}$). Thus, we propose an efficient algorithm below to solve problem (P$5$.a) by taking into account the cases with and without negative weights, respectively \cite{mei2020cooperative}. 
\begin{itemize}
\item Without negative weights:  In this case, classical Dijkstra algorithm can be directly applied to solve problem (P$5$.a).
\item With negative weights: If there exist negative weights, the classical Dijkstra algorithm is not  guaranteed to attain the optimal solution in general as it operates in a greedy manner. Thus, we resort to a recursive algorithm proposed in \cite{cheng2004finding} to obtain the shortest path of the SSPP, with the details omitted due to limited space.
\end{itemize}

\section{Select Active IRS for Beam Routing or not?}
\vspace{-3pt}
In this section, we compare the achievable rates of the cooperative  IRS beam routing designs with and without the active IRS involved, based on which we shed key insights on whether to select the active IRS for beam routing or not. 

First, consider the hybrid-IRS beam routing design with the active IRS. Let $f_{\rm BA}^*$ and $f_{\rm AU}^*$  denote its optimal channel power gain of the BS$\to$active IRS and active IRS$\to$user links, which are determined by the inter-IRS distances and number of passive reflecting elements. Then the corresponding achievable rate is given by 
$R^*_{\mathrm{act}}=\log_{2}(1+\gamma^*_{\mathrm{act}})$, where the maximum SNR is 
\begin{equation}\label{Eq:active}
\gamma^*_{\mathrm{act}}=\frac{P_{\mathrm{B}} N f_{\rm AU}^*  f_{\rm BA}^*}
{ f_{\rm AU}^* \sigma_{\mathrm{F}}^{2} + \frac{\sigma^{2}(P_\mathrm{B} f_{\rm BA}^*+\sigma_{\mathrm{F}}^{2})}{P_{\mathrm{F}}}}.
\vspace{-3pt}
\end{equation}
Next,  for the passive-IRS beam routing design, we denote by $\tilde{f}_{\rm BU}^*=|(\tilde{\boldsymbol{g}}_{\mathrm{BU}}^{H})^*  \boldsymbol{w}^*_{\rm B}|^2$ the maximum channel power gain \cite{mei2020cooperative}. Then the achievable rate is given by 
$R^*_{\mathrm{pas}}=\log_{2}(1+\gamma^*_{\mathrm{pas}})$, where its maximum SNR is 
%optimal SNR is formulated as
\begin{equation}\label{Eq:passive}
{\gamma^*_{\mathrm{pas}}=P_\mathrm{B} \tilde{f}_{\rm BU}^*} 
/{ \sigma^{2}}.
\end{equation}

Comparing \eqref{Eq:active} and \eqref{Eq:passive} leads to the following key result.

\begin{theorem}\label{The:Superior}
\emph{The active IRS should be selected for beam routing (i.e., $R^*_\mathrm{act}\ge R^*_\mathrm{pas}$), if we have 
\begin{align}
\frac{N}{\sigma_{\mathrm{F}}^{2}}  \geq \frac{\tilde{f}_{\rm BU}^*}{f_{\rm BA}^* \sigma^{2}} &+ 
\frac{P_\mathrm{B} \tilde{f}_{\rm BU}^*}{P_\mathrm{F} f_{\rm AU}^*\sigma_{\mathrm{F}}^{2}}  +\frac{\tilde{f}_{\rm BU}^*}{P_\mathrm{F} f_{\rm BA}^* f_{\rm AU}^*}
\end{align}}
\end{theorem}

Based on Theorem~\ref{The:Superior}, the effects of amplification power and number of reflecting elements of the active IRS on the active-IRS selection for beam routing are characterized as follows.

\begin{corollary}[Effect of amplification power]\label{Cor:Power}
\emph{Given the number of active reflecting elements $N$, we have $R_\mathrm{act}^*\ge R_\mathrm{pas}^*$, if the amplification power $P_{\mathrm{F}}$ satisfies
\begin{equation}
P_{\mathrm{F}} \geq \frac{\tilde{f}_{\rm BU}^* \sigma^{2} (P_\mathrm{B} f_{\rm BA}^* + \sigma_{\mathrm{F}}^{2}) }{f_{\rm AU}^*(N f_{\rm BA}^*\sigma^{2} - \tilde{f}_{\rm BU}^* \sigma_{\mathrm{F}}^{2})}.
\end{equation}}
\end{corollary}

Corollary~\ref{Cor:Power} shows that it is preferred to select the active IRS for cooperatively constructing a multi-reflection path if the active IRS's amplification power is sufficiently large. This is intuitively expected since a higher amplification power will lead to a higher amplification factor.

\begin{corollary}[Effect of number of active reflecting elements]\label{Cor:AE}
\emph{Given the amplification power $P_\mathrm{F}$, we have $R_\mathrm{act}^*\ge R_\mathrm{pas}^*$, if the number of active reflecting elements $N$ satisfies
\begin{align}
\!\!\! N \geq \frac{\tilde{f}_{\rm BU}^* \sigma_{\mathrm{F}}^{2}}{f_{\rm BA}^*\sigma^{2}} \!+\! 
\frac{P_\mathrm{B} \tilde{f}_{\rm BU}^*  }{P_\mathrm{F} f_{\rm AU}^*} 
+\frac{\tilde{f}_{\rm BU}^* \sigma_{\mathrm{F}}^{2}}{P_\mathrm{F} f_{\rm BA}^* f_{\rm AU}^*}.
\end{align}}
\end{corollary}

Corollary~\ref{Cor:AE} indicates that it is desirable to select the active IRS for beam routing if the number of active reflecting elements is sufficiently large. This is because the achievable rate over the hybrid IRS beam routing has a power scaling order of $\mathcal{O}(N)$.

\begin{remark}[What determines the beam routing path?]\label{Re:Affects}
\emph{
Lemma 2 shows that if the active IRS is involved in the beam routing, the best routing path is determined by the passive IRSs only via e.g., the number of passive reflecting elements $M$ and inter-IRS distance; while it is  independent of the parameters of active IRS. However, the optimal beam routing design in the proposed hybrid active/passive IRS aided system is jointly determined by the active and passive IRSs. Specifically, if the amplification power  $P_{\rm F}$ and/or number of active elements $N$ 
are sufficiently large, the active IRS should be selected for cooperative beam routing and vice versa.
%  and/or $Theorem~\ref{The:Superior} shows that the hybrid IRS beam routing tends to outperform the full-passive IRS beam routing when the amplification power is enough high and/or the number of active reflecting elements is large. This is in accordance with Corollary~\ref{Cor:Power} and Corollary~\ref{Cor:AE}.
 }       
\end{remark}
\vspace{-16pt}
\begin{figure}[h]
    \centering
\includegraphics[height=2cm,width=6cm]{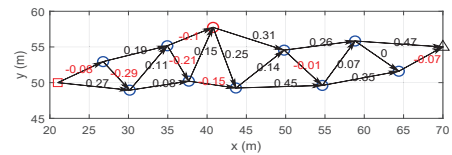}
    \vspace{-3pt}
    \caption{Graph representation of the simulation setup with $M=\mathrm{1400}$.}\label{fig:Graph}
    \vspace{-10pt}
\end{figure}
\vspace{-3pt}
\section{Numerical Results}
In this section, we present numerical results to demonstrate the efficacy of the proposed cooperative beam routing design for the new hybrid active/passive IRS aided wireless system. 
The simulation setup is as follows unless otherwise specified. As shown in Fig.~2,  one active IRS (red circle) and nine passive IRSs (blue circle) are deployed in  an indoor environment (e.g., smart factory), where the communication links with available LoS paths are represented by edges. Moreover, we set 
$M=1400$,
$T=4$, $\lambda=0.06$ m,  $\beta=(\lambda/4\pi)^{2} = -46$ dB,  $\sigma^{2}= -80$ dBm, and $\sigma_{\mathrm{F}}^{2} = -70$ dBm \cite{you2021wireless}.

First, we compare the beam routing designs of the proposed algorithm under different $N$ and $M$.
It is observed from Figs.~3(a)--3(b) that the number of passive reflecting elements  $M$ has significant effects on the beam routing path when the active IRS involves in the beam routing. Specifically, with a larger $M$ (i.e., $M=1500$), the multi-reflection path consists of more passive IRSs. This is because a larger $M$ leads to a higher passive beamforming gain, thus a multi-reflection path with more passive IRSs tends to provide more prominent multiplicative passive beamforming gain, which suffices to compensate the more severe multi-reflection path-loss. Besides, one can observe from  Figs.~3(b)--3(c) that when the number of active reflecting elements is sufficiently large, (i.e., $N \geq 800$ when $M=1200$), the hybrid-IRS beam routing design outperforms the passive-IRS counterpart, which agrees with the result in Corollary 2.

\begin{figure}[t]
    \centering
    \includegraphics[height=5cm,width=6.5cm]{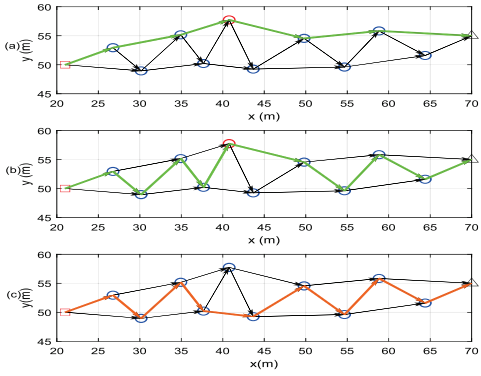}
    \caption{Optimal route versus  $M$ and $N$: (a) $M=800$, $N=300$; (b) $M=1200$, $N=800$; (c) $M= 1200$.}\label{fig:Routing}
    \vspace{-22pt}
\end{figure}

Next, we show in Fig.~4(a) the achievable rate of the proposed hybrid active/passive IRS beam routing design versus the amplification factor $P_{\rm F}$ with $M=N=2000$. For performance comparison, we consider the following three benchmark schemes:  1) \emph{hybrid IRSs with myopic beam routing}, which sequentially selects IRSs with the minimum edge weight for routing until reaching the user; 2) \emph{hybrid IRSs with random beam routing}, for which the IRSs are selected at random; 3) \emph{passive IRSs scheme}, for which the active IRS is replaced by a passive IRS and the corresponding beam routing design is obtained by similar methods in \cite{mei2020cooperative}. First, it is observed that the achievable rate of the hybrid IRS aided system monotonically increases with $P_{\rm F}$. Second, the proposed optimal beam routing scheme for the hybrid IRS aided system significantly outperforms the myopic and random benchmarks. Third,  the hybrid IRS aided system achieves a higher rate when $P_{\mathrm{F}}$ is sufficiently high (i.e., $P_{\mathrm{F}}\ge 14$ dBm), and the gain increases with $P_{\mathrm{F}}$. This is expected since a higher $P_{\mathrm{F}}$ tends to provide a higher power amplification gain, which is in accordance with Corollary~\ref{Cor:Power}.
  
Last, we show in Fig.~4(b) the achievable rate of the proposed hybrid IRS aided system versus $N$ with $M= 1200$. It is observed that the passive IRS aided system attains a higher rate when $N$ is small. However, when $N$ is sufficiently large (e.g., $N\geq500$ for $P_{\mathrm{F}}=20$ dBm, and $N\geq750$ for $P_{\mathrm{F}}=10$ dBm), the proposed hybrid-IRS aided system achieves a much larger rate than the passive-IRS system. This indicates that the optimal beam routing should incorporate the active IRS when $N$ is large, which is consistent with Corollary~\ref{Cor:AE}. 

\vspace{-0.3cm}
\section{Conclusions}
\vspace{-0.1cm}
In this paper, we considered a new hybrid active/passive IRS aided wireless communication system, where the active and passive IRSs cooperatively establish a virtual LoS multi-hop reflection path from the BS to the user. An efficient algorithm was proposed to design the beamforming of the BS and active/passive IRSs, as well as the multi-reflection path for rate maximization. 
It was shown that the proposed hybrid active/passive IRS beam routing design achieves a higher rate than the counterpart with passive IRSs only, when the amplification power and/or number of active reflecting elements are sufficiently large.

\begin{figure}[t]
\centering
\subfigure[Achievable rate versus amplification power, $P_{\mathrm{F}}$.]{\label{FigWeightRateW}
\includegraphics[height=3.2cm]{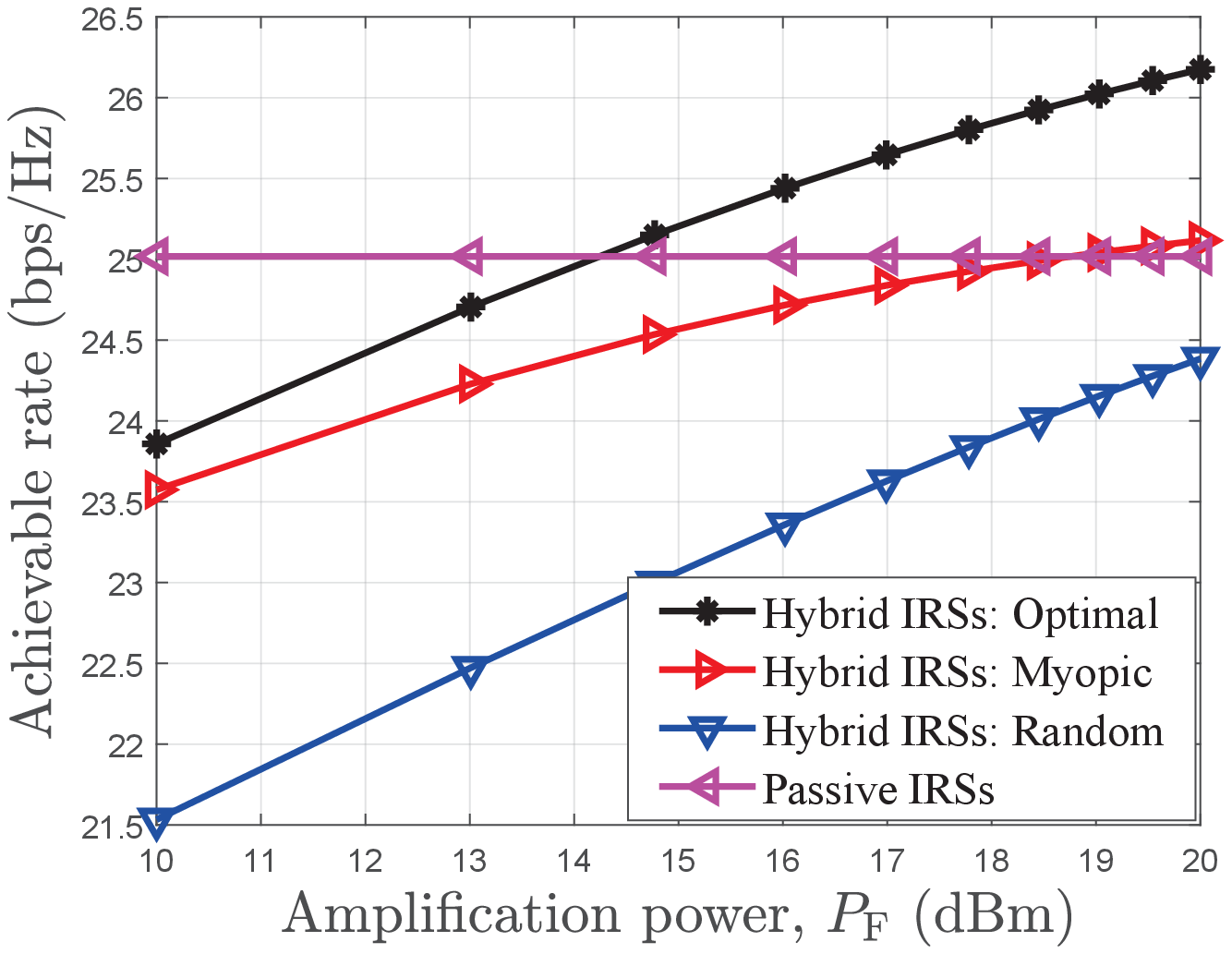}}
\hspace{2pt}
\subfigure[Achievable rate versus number of active IRS elements, $N$.]{\label{FigRateNUpDown}
\includegraphics[height=3.2cm]{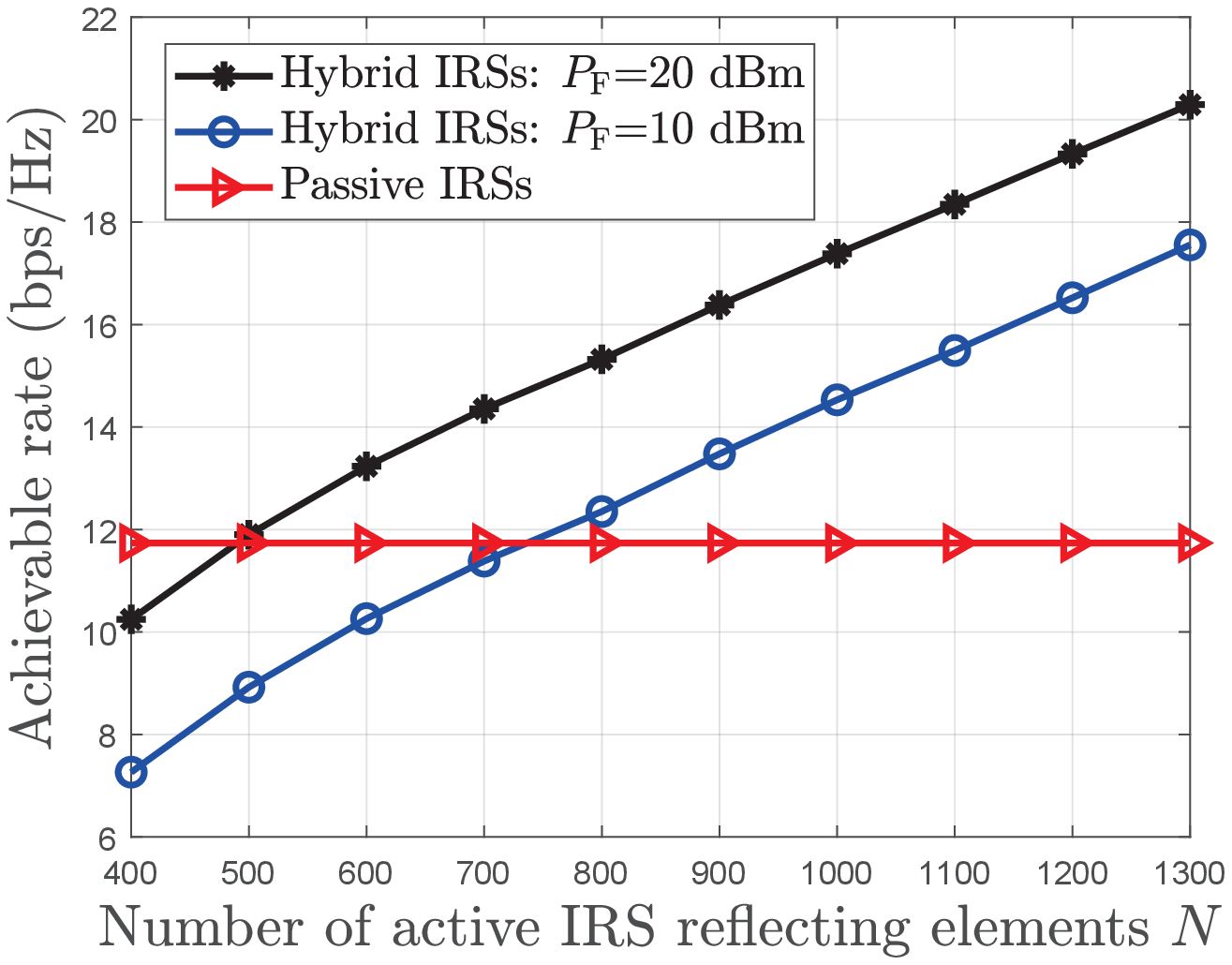}}
\label{Fig:Rate}
\vspace{-5pt}
\caption{Effects of system parameters on the achievable rate.}
\vspace{-18pt}
\end{figure} 

\bibliographystyle{IEEEtran}
\vspace{-10pt}
\bibliography{IEEEabrv,Ref}
\vspace{-6pt}

\end{document}